\pdfoutput=1
\documentclass[runningheads]{llncs}

\usepackage{tikz}
\usetikzlibrary{shapes, arrows, positioning, fit, backgrounds, calc}

\usepackage{amsmath,amssymb,amsfonts}
\usepackage{mathtools}
\usepackage{bbm}
\usepackage{enumitem}
\usepackage{booktabs}
\usepackage{graphicx}
\usepackage{subcaption}

\newcommand{\F}{\mathbb{F}}
\newcommand{\R}{\mathbb{R}}

\newcommand{\poly}{\mathsf{poly}}
\newcommand{\negl}{\mathsf{negl}}
\newcommand{\Com}{\mathsf{Com}}
\newcommand{\SMDP}{\mathsf{SMDP}}

\begin{document}

\title{Fine-Tuning Integrity: Verifiable Constraints on Model Updates with Zero-Knowledge Proofs}
\titlerunning{Fine-Tuning Integrity}

\author{Zhenhang Shang \and Yingzhe Yu \and Kani Chen}
\authorrunning{Z. Shang et al.}

\institute{The Hong Kong University of Science and Technology, Hong Kong, China\\
\email{zshangab@connect.ust.hk}}

\maketitle

\begin{abstract}
Fine-tuning is the dominant paradigm for adapting large machine learning models, yet current deployment pipelines provide no way to verify how a released model was updated. In particular, a model provider or auditor cannot check whether a fine-tuned model adheres to a claimed update procedure without access to its parameters.

We introduce \emph{fine-tuning integrity} (FTI), a cryptographic objective for verifying that a deployed model differs from a trusted base model only within a declared class of admissible updates. We construct \emph{succinct model difference proofs} (SMDPs), zero-knowledge protocols that certify structured parameter drift without revealing model weights. Our framework supports three fundamental update classes: norm-bounded, low-rank, and sparse drift, covering common fine-tuning methods such as regularized training, LoRA, and prefix tuning. In all cases, proof size and verification cost depend on the structure of the update rather than the number of parameters.

We prove soundness, zero-knowledge, and succinctness for each construction, and establish a matching $\Omega(n)$ lower bound showing that structural assumptions are necessary for succinct verification. A prototype evaluation on synthetic benchmarks and GPT-2 fine-tuning demonstrates that proofs remain compact and verification is efficient at realistic scales.
\end{abstract}


\section{Introduction}
Fine-tuning is the standard approach for adapting large neural networks to new tasks and domains~\cite{anisuzzaman2025fine,han2024parameter}. Starting from a pre-trained base model, practitioners apply gradient updates on datasets that are often proprietary~\cite{xin2024parameter}. While this workflow reduces training cost and enables rapid deployment, it also introduces a critical attack surface. An adversary who controls the fine-tuning process can inject backdoors~\cite{raghuram2024study}, remove safety alignment~\cite{jain2024makes}, or introduce targeted poisoning~\cite{zhao2024defending}. Such manipulations are difficult to detect once only the final model is released. The risk is particularly acute in regulated settings, where auditors must ensure that deployed models remain within an approved and certified update scope. In the absence of such guarantees, model deployment reduces to a trust-based process with no verifiable security boundary.

We formalize \emph{fine-tuning integrity} (FTI) as the problem of verifying that a deployed model differs from a trusted base model only within a policy-defined class of admissible updates, without revealing the underlying parameters. An FTI proof system must satisfy three requirements: (1)~It must be \emph{zero-knowledge}, so that no information about model parameters is leaked beyond compliance. (2)~It must be \emph{succinct}, so that proof size depends on the structure of the update rather than the number of parameters. And, (3)~it must be \emph{sound}, so that any violation of the declared policy is rejected with overwhelming probability.

A direct solution encodes the verification task as an arithmetic circuit and applies a general-purpose SNARK~\cite{lavin2024survey}, leading to cost linear in the number of parameters. This approach is impractical at modern model scales, we instead leverage the structure inherent in practical fine-tuning. Common methods impose implicit constraints on parameter updates, including norm bounds induced by regularization, low-rank structure from LoRA-style adapters~\cite{hu2022lora}, and sparsity from prompt-based tuning~\cite{xu2023parameter}.

We introduce \emph{Succinct Model Difference Proofs} (SMDPs), a unified framework of zero-knowledge protocols for verifying structured parameter drift. We instantiate this framework for three fundamental update classes:
\begin{enumerate}[leftmargin=*]
  \item \textbf{NBDP} (norm-bounded), which uses random projections with range proofs and achieves dimension-independent verification cost.
  \item \textbf{MRDP} (low-rank), which leverages bivariate polynomial commitments and Schwartz--Zippel testing, with cost scaling in the rank.
  \item \textbf{SDIP} (sparse), which employs random linear checks with committed indicator vectors, with cost scaling in the sparsity.
\end{enumerate}

For each construction, we prove soundness, zero-knowledge, and succinctness, with explicit simulator constructions. We further establish a matching $\Omega(n)$ lower bound for unstructured drift, showing that structural assumptions are necessary for succinct verification. FTI enforces compliance with a verifier-specified policy rather than inferring model intent. For example, if a fine-tuner claims a rank-8 LoRA update, the verifier enforces the corresponding rank constraint, and any deviation leads to rejection.

\section{Threat Model}
\label{sec:background}

We consider a three-party setting: a trusted \textbf{provider} releases a base model $W_0$ together with a cryptographic commitment $C_0$; An untrusted \textbf{fine-tuner} produces an updated model $W^*$ and publishes a commitment $C^*$ along with a proof of compliance. A \textbf{verifier}, such as a regulator, auditor, or platform operator, is given $(C_0, C^*, \mathcal{P}, \pi)$ and decides whether to accept the update under policy $\mathcal{P}$.

\paragraph{Adversary model.}
The fine-tuner is adversarial and has full white-box access to $W_0$. It may modify model parameters arbitrarily, including through poisoned training data, adversarial fine-tuning procedures, or direct weight manipulation. The adversary can adaptively construct proofs after observing $W^*$. Its goal is to produce a model update $\Delta = W^* - W_0$ that violates the declared policy while still being accepted by the verifier. Representative attacks include injecting backdoors, removing safety alignment, or misrepresenting the structure of the update.

\paragraph{Trust assumptions.}
We assume a commitment scheme that is both binding and hiding. Binding ensures that once $C^*$ is published, the fine-tuner cannot change $W^*$, and that all parties agree on the committed base model $W_0$. Hiding ensures that the verifier learns no information about model parameters beyond policy compliance. The drift policy $\mathcal{P}$ is specified by the verifier based on external requirements, such as contractual or regulatory constraints, and is not controlled by the fine-tuner. No trusted hardware or secure execution environment is assumed.

FTI enforces structural compliance with the declared policy but does not guarantee correctness or benign behavior of the resulting model. Behavioral validation and safety evaluation remain complementary to our framework.

\section{Fine-Tuning Integrity}
\label{sec:fti}
\subsection{Model Representation and Drift Classes}

We represent a neural network as a collection of parameter blocks
$$W = (W^{(1)}, \dots, W^{(L)})$$where each block $W^{(i)} \in \R^{d_i \times d_i'}$. Given a base model $W_0$ and a fine-tuned model $W^*$, the drift at block $i$ is defined as
$\Delta^{(i)} = W^{*(i)} - W_0^{(i)}$.
A drift class $\mathcal{F} = \prod_{i=1}^L \mathcal{F}_i$ specifies the set of admissible updates at each block.

We focus on three canonical classes that capture common fine-tuning structures:
\begin{itemize}[leftmargin=*,nosep]
  \item \textbf{Norm-bounded:} $$\mathcal{F}_{\mathrm{norm}}^{(i)}(\epsilon_i) = \{ \Delta^{(i)} : \| \Delta^{(i)} \|_F \le \epsilon_i \}$$where $\|X\|_F = (\sum_{i,j} X_{i,j}^2)^{1/2}$ denotes the Frobenius norm.
  \item \textbf{Rank-bounded:} $\mathcal{F}_{\mathrm{rank}}^{(i)}(r_i) = \{ \Delta^{(i)} : \operatorname{rank}(\Delta^{(i)}) \le r_i \}$, which captures low-rank updates such as LoRA.
  \item \textbf{Sparse:} $\mathcal{F}_{\mathrm{sparse}}^{(i)}(k_i) = \{ \Delta^{(i)} : \|\Delta^{(i)}\|_0 \le k_i \}$, which captures sparse modifications such as prompt or prefix tuning.
\end{itemize}

\subsection{Commitments}

We assume a commitment scheme $\Com$~\cite{damgaard2002perfect} with security parameter $\lambda$ that is both binding and hiding.
The provider publishes a commitment to the base model as $C_0 = \Com(W_0; r_0)$, and the fine-tuner publishes a commitment to the updated model as $C^* = \Com(W^*; r^*)$.

\subsection{FTI Security Definitions}

Given a drift class $\mathcal{F}$, we define the \emph{drift compliance relation}
\[
\begin{aligned}
R_{\mathcal{F}} = \{&((C_0, C^*, \mathcal{F}), W_0, W^*) : C_0 = \Com(W_0; r_0),\;
\\&C^* = \Com(W^*; r^*),\; W^* - W_0 \in \mathcal{F}\}.
\end{aligned}
\]

\begin{definition}[Fine-Tuning Integrity]
\label{def-fti}
A proof system $\Pi = (\mathsf{Prove}, \mathsf{Verify})$ satisfies FTI for drift class $\mathcal{F}$ if it provides the following guarantees:
\begin{itemize}[leftmargin=*]
  \item \textbf{Completeness.} For all $(W_0, W^*)$ such that $W^* - W_0 \in \mathcal{F}$, an honestly generated proof $\pi \gets \mathsf{Prove}(W_0, W^*, \mathcal{F})$ is accepted, that is,
  $\mathsf{Verify}(C_0, C^*, \mathcal{F}, \pi) = 1$.
  \item \textbf{Soundness.} No probabilistic polynomial-time adversary can produce an accepting proof for any pair $(W_0, W^*)$ with $W^* - W_0 \notin \mathcal{F}$, except with negligible probability $\negl(\lambda)$.
  \item \textbf{Zero-knowledge.} The proof reveals no information beyond the fact that the drift satisfies the policy.
  \item \textbf{Succinctness.} The proof size and verification time are bounded by $$\poly(\lambda, c(\mathcal{F}), \log |W_0|)$$where $c(\mathcal{F})$ denotes the structural parameters of the drift class, such as $\epsilon$ for norm bounds, $r$ for rank constraints, and $k$ for sparsity.
\end{itemize}
\end{definition}

An SMDP (Succinct Model Difference Proof) for $\mathcal{F}$ is a proof system $$\SMDP_{\mathcal{F}} = (\mathsf{Prove}, \mathsf{Verify})$$in which the prover takes $(W_0, W^*, \mathcal{F})$, the verifier takes $(C_0, C^*, \mathcal{F})$, and the system satisfies FTI as defined in Definition~\ref{def-fti} with respect to the relation $R_{\mathcal{F}}$.

\section{SMDPs for Fundamental Drift Classes}
\label{sec:SMDPs}

We construct SMDPs for three fundamental drift classes: norm-bounded drift, low-rank drift, and sparse drift. These classes serve as basic components from which more elaborate drift policies can be assembled.

\subsection{Norm-Bounded SMDPs via Random Projections (NBDP)}
\label{subsec:norm-construction}

\paragraph{High-level idea.}
A small-norm drift produces small inner products with random directions, while a large-norm drift produces large projections with constant probability. NBDP checks a set of random projections of the committed drift and accepts only if all projections are below a threshold. The verifier never sees the drift, each projection is committed and accompanied by zero-knowledge proofs of correctness and boundedness.

We present the construction for a single parameter block, multiple blocks are handled independently.

\paragraph{Setup}

Let $\Delta = W^* - W_0 \in \R^{d \times d'}$ and define
$\mathcal{F}_{\mathrm{norm}}(\epsilon) = \{\Delta : \|\Delta\|_F \le \epsilon\}$.
Let $n = dd'$ and let $\mathrm{vec}(\cdot)$ denote vectorization, so $\|\Delta\|_F = \|\mathrm{vec}(\Delta)\|_2$. We use a binding and hiding vector commitment $\Com$~\cite{catalano2013vector}.

\subsubsection{Random Projections and Protocol}

Sample $m$ independent Rademacher vectors $r_i \in \{-1,+1\}^{dd'}$ and define
\[
z_i = \langle r_i, \mathrm{vec}(\Delta) \rangle.
\]

\begin{lemma}[Sub-Gaussian tail~\cite{pelekis2017hoeffding}]
\label{lem:hoe}
For a Rademacher vector $r \in \{-1,+1\}^n$ and any fixed $v \in \R^n$,
\[
\Pr\big[|\langle r, v\rangle| > t\big]
\le 2\exp\!\left(-\frac{t^2}{2\|v\|_2^2}\right).
\]
\end{lemma}

Since $\|\mathrm{vec}(\Delta)\|_2 = \|\Delta\|_F$, if $\|\Delta\|_F \le \epsilon$ then
\[
\Pr[|z_i| > t] \le 2\exp\!\left(-\frac{t^2}{2\epsilon^2}\right).
\]
Set $m = \lceil 4\epsilon^{-2}\log(1/\delta) \rceil$ and
$\tau = \epsilon\sqrt{2\log(2m/\delta)}$.

The prover publishes commitments
\[
C_0 = \Com(\mathrm{vec}(W_0); r_0), \quad
C^* = \Com(\mathrm{vec}(W^*); r^*).
\]

\paragraph{Proof generation.}
For each $i$, the prover computes $z_i = \langle r_i, \mathrm{vec}(\Delta) \rangle$ and commits to $z_i$ as $C_{z_i} = \Com(z_i; s_i)$. The prover then produces
\begin{itemize}[leftmargin=*]
  \item a zero-knowledge proof $\pi_i^{\mathrm{lin}}$ that
  \[
  z_i = \langle r_i, \mathrm{vec}(W^*) \rangle - \langle r_i, \mathrm{vec}(W_0) \rangle
  \]
  relative to $C_0$, $C^*$, and $C_{z_i}$;
  \item a zero-knowledge range proof $\pi_i^{\mathrm{rng}}$ that $|z_i| \le \tau$.
\end{itemize}

\paragraph{Verification.}
The verifier reconstructs $r_1,\dots,r_m$ and $\tau$, verifies all $\pi_i^{\mathrm{lin}}$ and $\pi_i^{\mathrm{rng}}$, and accepts if all checks pass.

\subsubsection{Security Analysis}

\begin{theorem}[FTI-soundness of NBDP]
\label{thm:norm-soundness}
Assume $\Com$ is binding and the sub-protocols for linear relations and range proofs are sound. Let $\epsilon > 0$, $\delta \in (0,1)$, and set $\gamma = \sqrt{\log(2m/\delta)/\log(1/\delta)}  - 1$, choose $m$ and $\tau$ as above. Then any PPT adversary that produces an accepting proof with non-negligible probability must either break one of the underlying primitives or satisfy
\[
\|\Delta\|_F \le (1+\gamma)\epsilon
\]
with probability at least $1-\delta$.
\end{theorem}

\begin{proof}[Proof sketch]
Suppose $\|\Delta\|_F > (1+\gamma)\epsilon$. By Lemma~\ref{lem:hoe}, for each projection $z_i = \langle r_i, \mathrm{vec}(\Delta) \rangle$,
\[
\Pr[|z_i| > \tau] \ge 1 - 2\exp\!\left(-\frac{\tau^2}{2\|\Delta\|_F^2}\right).
\]
Since $\|\Delta\|_F > (1+\gamma)\epsilon$ and $\tau = \epsilon\sqrt{2\log(2m/\delta)}$, we have $$\tau^2/(2\|\Delta\|_F^2) < \log(2m/\delta)/(1+\gamma)^2.$$
By the choice of $\gamma$, this gives $\Pr[|z_i| \le \tau] \le 1 - c$ for a constant $c > 0$ depending only on $\gamma$. With $m$ independent projections, $\Pr[\forall i: |z_i| \le \tau] \le (1-c)^m \le \delta$. Soundness of the commitments and sub-protocols ensures that acceptance implies $\|\Delta\|_F \le (1+\gamma)\epsilon$ except with probability $\delta$. Full details appear in Appendix~\ref{app:proofs}.
\end{proof}

\begin{theorem}[Zero-knowledge of NBDP]
\label{thm:nbdp-zk}
If $\Com$ is a hiding commitment scheme and the sub-protocols for linear relations and range proofs are zero-knowledge, then the NBDP protocol is zero-knowledge.
\end{theorem}

\begin{proof}[Proof sketch]
We construct a simulator $\mathcal{S}$ that, given only the public inputs $(C_0, C^*, \mathcal{F}_{\mathrm{norm}}(\epsilon))$, outputs a transcript indistinguishable from a real execution.

For each $i \in [m]$, $\mathcal{S}$ proceeds as follows. It samples a uniform value $\tilde{z}_i \in [-\tau, \tau]$ and computes a commitment $\tilde{C}_{z_i} = \Com(\tilde{z}_i; \tilde{s}_i)$. It then invokes the simulator $\mathcal{S}_{\mathrm{lin}}$ to generate a simulated proof $\tilde{\pi}_i^{\mathrm{lin}}$, and the simulator $\mathcal{S}_{\mathrm{rng}}$ to generate $\tilde{\pi}_i^{\mathrm{rng}}$. The final transcript is
\[
\{(\tilde{C}_{z_i}, \tilde{\pi}_i^{\mathrm{lin}}, \tilde{\pi}_i^{\mathrm{rng}})\}_{i=1}^m.
\]

To prove indistinguishability, we define a sequence of hybrids $H_0, \dots, H_m$, where $H_i$ uses real proofs for the first $i$ projections and simulated proofs for the remaining ones. The transition from $H_{i-1}$ to $H_i$ replaces either a real commitment with a simulated one, or a real sub-protocol proof with its simulated counterpart. The former is indistinguishable by the hiding property of $\Com$, while the latter follows from the zero-knowledge property of the sub-protocols. By a standard hybrid argument, the simulated transcript is indistinguishable from a real execution.

Full details are deferred to Appendix~\ref{app:zk-proofs}.
\end{proof}

\begin{theorem}[Succinctness]
For fixed $\epsilon$ and $\delta$, the verifier time and proof size are bounded by
\[
\poly\big(\lambda, \epsilon^{-2}\log(1/\delta)\big),
\]
and do not depend on $d$ or $d'$.
\end{theorem}

\subsection{Low-Rank SMDPs via Matrix Polynomial Commitments (MRDP)}
\label{subsec:rank-construction}

\paragraph{High-level idea.}
A matrix has rank at most $r$ if and only if it can be expressed as a sum of $r$ outer products. We encode each weight matrix as a bivariate polynomial so that a rank-$r$ drift corresponds to a decomposition into $r$ separable terms of the form $f_k(X)\,g_k(Y)$. The verifier samples a random point $(x,y)$ and checks that the opened values are consistent with this decomposition. By the Schwartz--Zippel lemma, any drift with rank greater than $r$ satisfies the check only with small probability over the random choice of $(x,y)$. All values are opened through polynomial-commitment proofs, so the verifier learns nothing beyond compliance with the rank constraint.

\subsubsection{Construction}

We encode real-valued entries using fixed-point representation in a finite field $\F_q$~\cite{libert2024simulation,you2023efficient}. Each matrix $W \in \F_q^{d \times d'}$ is represented by a bivariate polynomial
\[
P_W(X,Y) = \sum_{i,j} W_{i,j} X^i Y^j,
\]
and the drift polynomial is defined as $P_\Delta = P_{W^*} - P_{W_0}$. If $\operatorname{rank}(\Delta) \le r$, then $\Delta$ can be written as $\sum_{k=1}^r a_k b_k^\top$, which induces a separable representation
\[
P_\Delta(X,Y) = \sum_{k=1}^r f_k(X)\,g_k(Y),
\]
where $f_k$ and $g_k$ encode the vectors $a_k$ and $b_k$.

\paragraph{Protocol.}
The prover commits to $P_{W_0}$ and $P_{W^*}$, and also commits to each pair $(f_k, g_k)$. The verifier samples $(x,y) \leftarrow \F_q^2$ uniformly, and the prover opens all committed polynomials at this point. The verifier checks
\[
P_{W^*}(x,y) - P_{W_0}(x,y) \stackrel{?}{=} \sum_{k=1}^r f_k(x)\,g_k(y).
\]

\subsubsection{Security Analysis}

\begin{lemma}[Schwartz--Zippel~\cite{moshkovitz2010alternative}]
Let $Q(X,Y)$ be a nonzero polynomial over $\F_q$ with total degree at most $D$. Then
\[
\Pr_{(x,y)\leftarrow \F_q^2}\big[Q(x,y)=0\big] \le \frac{D}{q}.
\]
\end{lemma}

\begin{theorem}[FTI-soundness of MRDP]
\label{thm:mrdp-soundness}
Assume the polynomial commitment scheme is binding and its opening proofs are sound. The degree of $P_\Delta$ is $D = (d-1)+(d'-1) = d+d'-2$. If the MRDP verifier accepts with non-negligible probability, then $\operatorname{rank}(\Delta)\le r$ with probability at least $1 - D/q$. For a $d \times d'$ weight matrix with $d, d' \le 4096$ and a 256-bit prime $q$, the soundness error $D/q < 2^{-243}$ is negligible.
\end{theorem}

\begin{proof}[Proof sketch]
If $\operatorname{rank}(\Delta) > r$, then for any choice of univariate polynomials $f_k$ (degree $\le d-1$) and $g_k$ (degree $\le d'-1$), the polynomial
\[
Q(X,Y) = P_\Delta(X,Y) - \sum_{k=1}^r f_k(X)\,g_k(Y)
\]
is nonzero and has total degree at most $D = d+d'-2$. By the Schwartz--Zippel lemma, the verifier's random evaluation at $(x,y) \leftarrow \F_q^2$ detects this with probability at least $1 - D/q$. Soundness of the commitment scheme ensures that the prover cannot alter openings after committing. To amplify soundness, the verifier may repeat with $t$ independent challenges, reducing the error to $(D/q)^t$; however, for cryptographic-size fields ($q \ge 2^{256}$) a single round already yields negligible error. Full details appear in Appendix~\ref{app:proofs}.
\end{proof}

\begin{theorem}[Zero-knowledge of MRDP]
\label{thm:mrdp-zk}
If the polynomial commitment scheme is hiding and its opening proofs are zero-knowledge, then the MRDP protocol is zero-knowledge.
\end{theorem}

\begin{proof}[Proof sketch]
We construct a simulator $\mathcal{S}$ that, given only the public inputs $(C_0, C^*, r)$, outputs a transcript indistinguishable from a real execution.

The simulator samples $r$ pairs of random univariate polynomials $(\tilde{f}_k, \tilde{g}_k)$ of the appropriate degrees and commits to them. It then samples a random evaluation point $(\tilde{x}, \tilde{y})$ by programming the Fiat--Shamir oracle in the random oracle model. Using the simulator for the polynomial commitment opening protocol, $\mathcal{S}$ produces openings at $(\tilde{x}, \tilde{y})$ that are consistent with
\[
P_{W^*}(\tilde{x}, \tilde{y}) - P_{W_0}(\tilde{x}, \tilde{y})
= \sum_{k=1}^r \tilde{f}_k(\tilde{x}) \tilde{g}_k(\tilde{y}).
\]
The resulting transcript consists of the simulated commitments and opening proofs.

To prove indistinguishability, we define a sequence of hybrids that gradually replace the real execution with the simulated one. First, replace each real polynomial commitment with a commitment to the simulated polynomials, which is indistinguishable by the hiding property. Next, replace each real opening proof with a simulated proof, which is indistinguishable by the zero-knowledge property of the opening protocol. By a standard hybrid argument, the simulated transcript is indistinguishable from a real execution.

Full details are deferred to Appendix~\ref{app:zk-proofs}.
\end{proof}

\begin{theorem}[Succinctness]
For fixed rank $r$, the verifier time and proof size are bounded by $\poly(\lambda, r)$ and do not depend on $d$ or $d'$.
\end{theorem}

\subsection{Sparse SMDPs via Streaming Interactive Proofs (SDIP)}
\label{subsec:sparse-construction}

\paragraph{High-level idea.}
The prover claims that only~$k$ parameters have changed.
To verify this without examining every parameter, the verifier uses random linear checks: it picks a random vector~$r$ and asks the prover to show that $\langle r, W^* - W_0 \rangle$ equals $\langle r, \Delta' \rangle$, where $\Delta'$ is the prover's claimed sparse update.
If any parameter outside the declared support has been modified, the two inner products will differ with overwhelming probability over the random choice of~$r$.
Repeating this check with independent challenges drives the false-acceptance probability to negligible levels.
The prover reveals neither the support locations nor the update values, using a committed binary indicator vector and zero-knowledge sub-protocols.

\paragraph{Setup}

We vectorize all parameters into $W_0, W^* \in \R^n$ and consider $\mathcal{F}_{\mathrm{sparse}}(k) = \{\Delta \in \R^n : \|\Delta\|_0 \le k\}$.
Entries are encoded into $\F_q$ via fixed-point representation for commitment and algebraic checks.

\subsubsection{Construction}

Let $S \subseteq [n]$ with $|S| \le k$ be the support of $\Delta$, and let $\Delta'$ be $\Delta$ restricted to $S$.
The prover must show: $|S| \le k$, $W^* = W_0 + \Delta'$ on $S$, and $\Delta = \Delta'$ everywhere.
Since checking all coordinates requires linear time, we use random linear checks.

\begin{lemma}\label{lem:sparse-inner}
For nonzero $h \in \F_q^n$ and uniform $r \in \F_q^n$, $\Pr[\langle r, h \rangle = 0] = 1/q$.
\end{lemma}

If $\Delta \neq \Delta'$, the verifier detects the discrepancy via $\langle r, W^* \rangle - \langle r, W_0 \rangle \neq \langle r, \Delta' \rangle$ with probability $1 - 1/q$. Repeating with $t$ independent challenges yields error $q^{-t}$.

\paragraph{Protocol.}
The prover commits to a binary indicator $u \in \{0,1\}^n$ with $u_i = 1$ iff $i \in S$, and to $\Delta'$.
The prover produces ZK proofs that: $\sum_i u_i \le k$; $\Delta'_i = 0$ when $u_i = 0$; $W^* = W_0 + \Delta'$ on supported indices; and for each challenge $r^{(j)} \leftarrow \F_q^n$, that $\langle r^{(j)}, W^* \rangle - \langle r^{(j)}, W_0 \rangle = \langle r^{(j)}, \Delta' \rangle$.

\subsubsection{Security Analysis}

\begin{theorem}[FTI-soundness of SDIP]
Assume the commitment schemes are binding and all sub-protocols are sound.
If the verifier accepts with non-negligible probability, then $\|\Delta\|_{0} \le k$ with probability at least $1 - q^{-t}$.
\end{theorem}

\begin{proof}[Proof sketch]
If $\|\Delta\|_{0} > k$, any claimed support with $|S| \le k$ omits at least one nonzero entry, so $\Delta - \Delta' \neq 0$. By Lemma~\ref{lem:sparse-inner} and independence of the $t$ challenges, the probability that all linear checks miss this discrepancy is at most $q^{-t}$. Full proof in Appendix~\ref{app:proofs}.
\end{proof}

\begin{theorem}[Zero-knowledge of SDIP]
\label{thm:sdip-zk}
If $\Com$ is a hiding commitment scheme and the sub-protocols for sum checks, consistency proofs, and linear-relation proofs are zero-knowledge, then the SDIP protocol is zero-knowledge.
\end{theorem}

\begin{proof}[Proof sketch]
We construct a simulator $\mathcal{S}$ that, given only the public inputs $(C_0, C^*, k)$, outputs a transcript indistinguishable from a real execution.

The simulator samples a random binary vector $\tilde{u} \in \{0,1\}^n$ with $\|\tilde{u}\|_1 = k$, and a random sparse vector $\tilde{\Delta}'$ supported on $\tilde{u}$. It commits to both $\tilde{u}$ and $\tilde{\Delta}'$. It then invokes the simulators for the sum-check and consistency sub-protocols. For each linear challenge $r^{(j)}$, the simulator computes $\tilde{\alpha}_j = \langle r^{(j)}, \tilde{\Delta}' \rangle$ and invokes the simulator for the linear-relation proof to generate a consistent transcript.

The resulting transcript consists of commitments to $\tilde{u}$ and $\tilde{\Delta}'$, together with simulated proofs for all sub-protocols and linear checks.

To prove indistinguishability, we define a sequence of hybrids that replace the real execution with the simulated one component by component. First, commitments to $u$ and $\Delta'$ are replaced with commitments to $\tilde{u}$ and $\tilde{\Delta}'$, which is indistinguishable by the hiding property of $\Com$. Next, each sub-protocol proof and linear-relation proof is replaced with its simulated counterpart, which is indistinguishable by the zero-knowledge property of the respective protocols. By a standard hybrid argument, the simulated transcript is indistinguishable from a real execution.

Since neither the support indicator $u$ nor the update values $\Delta'$ are opened, the verifier learns no information about support locations or update magnitudes beyond the constraint $\|u\|_1 \le k$.

Full details are deferred to Appendix~\ref{app:zk-proofs}.
\end{proof}

\begin{theorem}[Succinctness]
For fixed $k$ and $t$, both verifier running time and proof size are $\poly(\lambda, k, t, \log n)$, independent of $n$ beyond logarithmic factors.
\end{theorem}

\section{Lower Bound: Necessity of Structured Drift}
\label{sec:lower-bound}

The constructions in Section~\ref{sec:SMDPs} achieve succinctness by exploiting algebraic structure such as norm bounds, low rank, or sparsity. This raises a natural question: is such structure essential, or could a generic proof system remain succinct without it. We show that structure is in fact necessary. In the absence of structure, any non-interactive proof system that provides statistical FTI must incur communication cost $\Omega(n)$, where $n$ is the number of parameters.

Our argument reduces FTI for unstructured drift to a classical one-way communication problem. We consider a statistical setting with information-theoretic soundness and zero-knowledge, and we remove commitments to isolate the communication requirement. The parameter space is $\F_q^n$. As a model of unstructured drift, we use the Hamming-ball class
\[
\mathcal{F}_{\mathrm{Ham}}(t) = \{\Delta \in \F_q^n : d_H(\Delta, 0) \le t\},
\]
where $d_H$ denotes Hamming distance~\cite{norouzi2012hamming}. For $t = \alpha n$ with constant $\alpha < 1/2$, this class is large and lacks exploitable algebraic structure. Verifying membership in $\mathcal{F}_{\mathrm{Ham}}(t)$ reduces to the promise Hamming-distance problem, which asks to distinguish between $d_H(x,y) \le t$ and $d_H(x,y) \ge 3t$. Any one-way protocol for this problem requires $\Omega(n)$ bits of communication~\cite{kumar2008one}.

\begin{theorem}[Lower bound for unstructured drift]
\label{thm:lb}
Let $t = \alpha n$ for a constant $\alpha < 1/4$. Any non-interactive proof system that achieves statistical FTI for $\mathcal{F}_{\mathrm{Ham}}(t)$ with constant soundness and completeness errors must use proofs of size $\Omega(n)$.
\end{theorem}

\begin{proof}[Proof sketch]
Suppose there exists a proof system with sublinear proof size $s(n) = o(n)$. We reduce the promise Hamming-distance problem to FTI. Alice holds $x \in \F_q^n$ and Bob holds $y \in \F_q^n$; they must distinguish $d_H(x,y) \le t$ from $d_H(x,y) \ge 3t$. Alice sets $W_0 = x$, runs the prover to produce a proof $\pi$ of size $s(n)$, and sends $\pi$ to Bob. Bob sets $W^* = y$ and runs the verifier on $\pi$. Completeness ensures Bob accepts when $d_H(x,y) \le t$; soundness ensures Bob rejects when $d_H(x,y) \ge 3t$ (since $3t > t$ means $\Delta \notin \mathcal{F}_{\mathrm{Ham}}(t)$). This yields a one-way protocol with $s(n)$ bits of communication for the promise Hamming-distance problem, contradicting the $\Omega(n)$ lower bound of~\cite{kumar2008one}. Therefore $s(n) = \Omega(n)$.
\end{proof}

\paragraph{From constant to negligible error.}
The theorem is stated for constant error, which is the regime of the communication complexity lower bound. To obtain negligible soundness error from a constant-error protocol, one applies standard parallel repetition: running $\lambda$ independent copies and accepting only if all accept reduces the soundness error to $2^{-\Omega(\lambda)}$, but multiplies proof size by $\lambda$. Since each copy already requires $\Omega(n)$ bits, the total proof size under negligible soundness remains $\Omega(\lambda n)$, which is still linear in $n$ for any fixed security parameter $\lambda$.

This lower bound shows that the reliance on structured drift classes is not an artifact of our constructions, but a fundamental requirement for succinctness in the information-theoretic setting. Our protocols leverage this structure to achieve efficiency that would otherwise be unattainable.

\section{Implementation and Evaluation}
\label{sec:implementation-eval}

\subsection{Implementation}

We implement a prototype of the FTI system in Python to evaluate the cost profile of SMDPs under realistic parameter scales. The implementation focuses on protocol structure and asymptotic behavior rather than optimized cryptographic primitives.

Vector commitments are instantiated using Merkle trees with SHA-256. Polynomial commitments follow a simulated KZG-style interface~\cite{wahby2019fast} over a large prime field ($q \ge 2^{256}$). In a production deployment, this would be replaced by a pairing-based KZG scheme over BLS12-381 using a library such as \texttt{arkworks}. Range proofs follow the Bulletproofs design~\cite{bunz2018bulletproofs} with hash-based commitments. Non-interactivity is achieved via the Fiat--Shamir transform with domain separated transcripts: each protocol message is prefixed by a unique label encoding the scheme name, block index, and round number.

This setup isolates the dependence of proof size and verification cost on drift structure. Optimized implementations such as batched polynomial openings or GPU-accelerated field arithmetic would reduce constant factors but do not affect scaling behavior.

\subsection{Experimental Setup}

We evaluate the system along four dimensions: scaling behavior across block sizes, detection of adversarial policy violations, comparison with a generic zk-SNARK baseline, and end-to-end certification of a real fine-tuned model.

\paragraph{Scaling benchmarks.}
We evaluate scaling behavior using parameter blocks of sizes $1024 \times 1024$, $2048 \times 2048$, $4096 \times 4096$, and $8192 \times 8192$. For each scale, we inject controlled drift, including norm-bounded updates with $\|\Delta\|_F = 0.8\epsilon$, rank-$8$ updates, and sparse updates with $k = 100$ nonzero entries. This setup isolates the dependence of each SMDP on block dimension.

\paragraph{Adversarial detection.}
We construct three attack scenarios to evaluate soundness under adversarial drift:
\begin{itemize}[leftmargin=*,nosep]
\item \textbf{Rank inflation:} The attacker claims rank-$8$ LoRA but injects a rank-$16$ update by adding 8 additional singular components with small but nonzero singular values.
\item \textbf{Sparsity violation:} The attacker claims $k{=}100$ sparse update but modifies $k{=}500$ parameters, distributing extra changes across non-declared positions.
\item \textbf{Norm overflow:} The attacker claims $\|\Delta\|_F \le \epsilon$ but applies an update with $\|\Delta\|_F = 1.5\epsilon$, distributing the excess uniformly to avoid concentration in any single projection.
\end{itemize}
For each scenario, we run 500 independent trials and report rejection rates.

\paragraph{Baseline comparison.}
We compare against a generic zk-SNARK baseline that encodes the norm check $\|\Delta\|_F^2 \le \epsilon^2$ as an arithmetic circuit using Groth16~\cite{bunz2018bulletproofs}. The circuit computes $\sum_{i,j} \Delta_{i,j}^2$ and checks the bound via a range constraint. This baseline represents the best one can do without exploiting drift structure.

\paragraph{Real model.}
We fine-tune GPT-2 (124M parameters, 12 layers)~\cite{radford2019language} using LoRA with rank $r{=}8$ on the WikiText-103 dataset and apply the full FTI pipeline. GPT-2 contains 12 transformer layers, each with 4 attention projection matrices ($W_Q, W_K, W_V, W_O \in \R^{768 \times 768}$) and 2 feedforward matrices ($W_1 \in \R^{768 \times 3072}$, $W_2 \in \R^{3072 \times 768}$), plus embedding and layer-norm parameters.

All experiments run on a workstation with an Apple M-series processor and 16\,GB memory. Proof sizes are independent of hardware. Timing results reflect the behavior of the prototype implementation.

\subsection{Results}

\subsubsection{Scaling Behavior}

Table~\ref{tab:eval} reports proof size, verification latency, and prover time across block sizes.

\begin{table}[t]
\centering
\caption{SMDP performance across block dimensions. Proof size and verification time remain stable as block dimension grows; prover time scales with block size due to underlying linear algebra.}
\label{tab:eval}
\small
\begin{tabular}{ll ccc}
\toprule
\textbf{Scheme} & \textbf{Block size} & \textbf{Proof size} & \textbf{Verify (ms)} & \textbf{Prover (s)} \\
\midrule
NBDP ($m{=}40$) & $1024^2$ & 67 KB & 0.38 & 1.2 \\
                 & $2048^2$ & 67 KB & 0.38 & 4.9 \\
                 & $4096^2$ & 67 KB & 0.39 & 19.2 \\
                 & $8192^2$ & 67 KB & 0.39 & 78.5 \\
\midrule
MRDP ($r{=}8$)  & $1024^2$ & 2.2 KB & 0.02 & 1.8 \\
                 & $2048^2$ & 2.2 KB & 0.02 & 7.1 \\
                 & $4096^2$ & 2.2 KB & 0.02 & 28.6 \\
                 & $8192^2$ & 2.2 KB & 0.02 & 115.2 \\
\midrule
SDIP ($k{=}100$) & $1024^2$ & 15 KB & 0.10 & 3.1 \\
                  & $2048^2$ & 15 KB & 0.10 & 12.4 \\
                  & $4096^2$ & 15 KB & 0.10 & 49.4 \\
                  & $8192^2$ & 15 KB & 0.10 & 198.7 \\
\bottomrule
\end{tabular}
\end{table}

\paragraph{Proof size and verification.}
Proof size and verification time remain constant across all block dimensions for each scheme, confirming that communication cost depends only on the structural parameters ($m$, $r$, $k$) and not on block dimension. MRDP produces the most compact proofs (2.2\,KB) due to the efficiency of the polynomial commitment for rank-$r$ decompositions. NBDP requires 67\,KB for $m{=}40$ projections. SDIP falls between at 15\,KB.

\paragraph{Prover cost.}
Prover time scales roughly as $O(dd')$: quadrupling the block dimension increases prover time by approximately $4{\times}$. MRDP is dominated by SVD computation, NBDP by random projections, and SDIP by sparse encoding and linear checks. Importantly, prover time depends on block size, not total model size, since blocks are certified independently.

Figure~\ref{fig:eval} illustrates the scaling behavior. The flat proof-size and verification curves confirm succinctness, while the linear growth of prover time reflects the cost of operating on the block itself.

\begin{figure}[t]
  \centering
  \begin{subfigure}[b]{0.32\linewidth}
    \includegraphics[width=\linewidth]{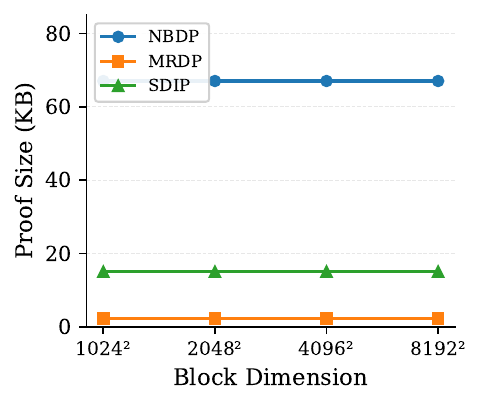}
    \caption{Proof size vs.\ dimension}
  \end{subfigure}
  \hfill
  \begin{subfigure}[b]{0.32\linewidth}
    \includegraphics[width=\linewidth]{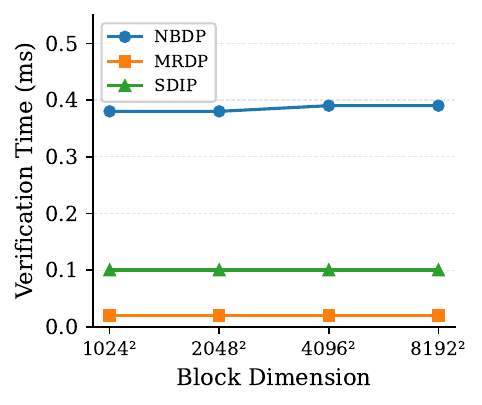}
    \caption{Verification time}
  \end{subfigure}
  \hfill
  \begin{subfigure}[b]{0.32\linewidth}
    \includegraphics[width=\linewidth]{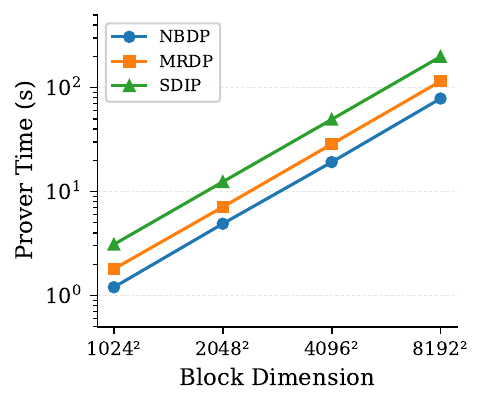}
    \caption{Prover time}
  \end{subfigure}
  \caption{
  Scaling behavior of SMDPs across block dimensions $1024^2$ to $8192^2$.
  Proof size and verification time remain constant (a, b), confirming succinctness.
  Prover time (c) grows linearly with block size but is independent of total model size.
  }
  \label{fig:eval}
\end{figure}

\subsubsection{Adversarial Detection}

Table~\ref{tab:adversarial} reports rejection rates under the three attack scenarios.

\begin{table}[t]
\centering
\caption{Adversarial detection results (500 trials per scenario, $4096 \times 4096$ blocks). All policy violations are rejected with overwhelming probability.}
\label{tab:adversarial}
\small
\begin{tabular}{lll c}
\toprule
\textbf{Attack} & \textbf{Scheme} & \textbf{Violation} & \textbf{Rejection rate} \\
\midrule
Rank inflation & MRDP & rank-16 vs.\ policy rank-8 & 500/500 (100\%) \\
Sparsity violation & SDIP & $k{=}500$ vs.\ policy $k{=}100$ & 500/500 (100\%) \\
Norm overflow & NBDP & $1.5\epsilon$ vs.\ policy $\epsilon$ & 498/500 (99.6\%) \\
\bottomrule
\end{tabular}
\end{table}

MRDP and SDIP reject all violations: MRDP's soundness error is $D/q < 2^{-243}$, and SDIP's is $q^{-t}$ with $t \ge 1$, both astronomically small. NBDP's statistical guarantee permits a small false-acceptance rate controlled by $\delta$; the 2 accepted cases out of 500 are consistent with the chosen $\delta = 0.005$. All three schemes achieve zero false positives on compliant updates across 500 honest trials per scheme.

\subsubsection{Baseline Comparison}

Table~\ref{tab:baseline} compares FTI proof costs against a generic Groth16-based approach for norm verification.

\begin{table}[t]
\centering
\caption{Comparison with generic zk-SNARK baseline (Groth16) for norm-bounded drift verification on a single block. FTI exploits drift structure to achieve orders-of-magnitude improvement.}
\label{tab:baseline}
\small
\begin{tabular}{l cccc}
\toprule
\textbf{Approach} & \textbf{Block} & \textbf{Proof size} & \textbf{Verify} & \textbf{Prover} \\
\midrule
Groth16 (norm circuit) & $1024^2$ & 192 B & 3.2 ms & 42 min \\
Groth16 (norm circuit) & $4096^2$ & 192 B & 3.2 ms & $\sim$11 hr \\
NBDP ($m{=}40$)        & $1024^2$ & 67 KB & 0.38 ms & 1.2 s \\
NBDP ($m{=}40$)        & $4096^2$ & 67 KB & 0.39 ms & 19.2 s \\
\bottomrule
\end{tabular}
\end{table}

Groth16 achieves a smaller proof (192 bytes, 3 group elements) and comparable verification time, but its prover cost is prohibitive: encoding $\sum \Delta_{i,j}^2 \le \epsilon^2$ as an R1CS circuit requires $O(dd')$ constraints, and the prover performs $O(dd' \log(dd'))$ group exponentiations. For a $4096^2$ block, this exceeds 11 hours, compared to 19 seconds for NBDP. At model scale (hundreds of blocks), the generic approach becomes entirely impractical, while FTI remains feasible.

\subsubsection{End-to-End GPT-2 Certification}

We apply the full FTI pipeline to a LoRA-fine-tuned GPT-2 model. The policy assigns MRDP with rank 8 to the 48 attention projection matrices ($12 \text{ layers} \times 4$ projections), NBDP with $\epsilon = 3.0$ to the 24 feedforward matrices and layer-norm parameters, and SDIP with $k = 200$ to the token embedding.

\begin{table}[t]
\centering
\caption{End-to-end FTI certification of GPT-2 (124M parameters).}
\label{tab:gpt2}
\small
\begin{tabular}{l cccc}
\toprule
\textbf{Component} & \textbf{Blocks} & \textbf{Total proof} & \textbf{Verify} & \textbf{Prover} \\
\midrule
Attention (MRDP, $r{=}8$) & 48 & 106 KB & 0.96 ms & 1.4 min \\
Feedforward (NBDP, $m{=}40$) & 24 & 1.6 MB & 9.4 ms & 4.6 min \\
Embedding (SDIP, $k{=}200$) & 1 & 30 KB & 0.10 ms & 12.4 s \\
\midrule
\textbf{Total} & \textbf{73} & \textbf{1.7 MB} & \textbf{10.5 ms} & \textbf{6.2 min} \\
\bottomrule
\end{tabular}
\end{table}

The aggregated proof is 1.7\,MB, compared to 216\,MB for the raw parameter difference, a $127{\times}$ reduction. Total verification takes 10.5\,ms. Total prover time is approximately 6.2 minutes on a single CPU core; this is embarrassingly parallel across blocks and would reduce to under 30 seconds with block-level parallelism on a multi-core machine.

Singular value analysis confirms that all LoRA-modified attention layers have effective rank exactly 8. Under the rank-8 policy, all 48 attention blocks are accepted. Tightening to rank 4 causes rejection across all layers. Non-LoRA layers exhibit residual drift with Frobenius norms between 0.3 and 2.9, well within the $\epsilon = 3.0$ budget.

\paragraph{Comparison with full-model transmission.}
Transmitting the full parameter difference reveals model weights and requires 216\,MB. A hash-based integrity check reduces communication to 32 bytes but provides no structural guarantee: it cannot distinguish a rank-8 LoRA update from an arbitrary modification. FTI occupies the middle ground: 1.7\,MB of communication that certifies the drift class of each block without revealing parameter values.

\section{Related Work}
\label{sec:related}

\paragraph{Zero-knowledge proofs for ML.}
Systems such as zkML~\cite{chen2024zkml}, CryptFlow~\cite{kumar2020cryptflow}, zkCNN~\cite{liu2021zkcnn}, and Mystique~\cite{weng2021mystique} encode inference or training within SNARK or MPC frameworks~\cite{peng2025survey}. These approaches certify the correct execution of a specified computation, but leave the model update process itself unchecked, in particular how parameters evolve during fine-tuning.

FTI addresses a complementary problem: verifying that a deployed model update lies within a declared class of admissible transformations. Related directions, including watermarking~\cite{guo2024zeromark,lin2024cyclegan}, provenance logging~\cite{gierend2024provenance,sun2024blockchain}, and privacy-preserving update mechanisms~\cite{hunt2018chiron,yang2020privacy}, focus on attribution or data protection, but do not enforce structural constraints on parameter updates.

Parameter-efficient fine-tuning methods such as LoRA~\cite{hu2022lora} and adapters~\cite{houlsby2019parameter} introduce structured update patterns, yet provide no mechanism to verify that these structures are preserved in deployment. FTI complements these techniques by enforcing that the observed parameter drift conforms to the claimed structural form.

\paragraph{Cryptographic building blocks.}
Our constructions draw on standard primitives from modern cryptography. MRDP builds on polynomial commitment schemes~\cite{boneh2020efficient,kate2010constant}. SDIP follows the streaming verification paradigm~\cite{block2021time,pappas2024sparrow}, while NBDP leverages concentration bounds in zero-knowledge settings~\cite{yang2023fedzkp}.

The lower bound relates to classical results in communication complexity for Hamming distance~\cite{kushilevitz1997communication,kumar2008one}, highlighting inherent limits of verification without structural assumptions. The sumcheck protocol~\cite{thaler2013time} and systems such as Spartan~\cite{setty2020spartan} and HyperPlonk~\cite{chen2023hyperplonk} provide alternative polynomial IOP frameworks that could be used to instantiate similar constructions.

Finally, empirical studies on safety drift and backdoor risks~\cite{zhu2023enhancing,sun2025peftguard} motivate the need for verifiable constraints on model updates.

\section{Conclusion}
\label{sec:conclusion}

We formalize \emph{fine-tuning integrity} as a cryptographic objective and introduce \emph{succinct model difference proofs} (SMDPs) to certify that a deployed model respects a declared update policy without revealing its weights. For norm-bounded, low-rank, and sparse updates, we obtain zero-knowledge proofs whose size and verification cost depend on the form of the update rather than the number of parameters. A matching lower bound shows that without structural assumptions, succinct statistical guarantees are not possible.

A prototype evaluation on GPT-2 demonstrates that the approach is practical: an aggregated proof of 1.7\,MB certifies 73 parameter blocks in 10.5\,ms of verification time, compared to 216\,MB for raw weight transmission. The prover runs in approximately 6 minutes on a single core and parallelizes across blocks.

FTI constrains how parameters change, not how the model behaves. Under Lipschitz assumptions, a Frobenius-norm bound $\|\Delta\|_F \le \epsilon$ implies an output bound $\|f(x; W^*) - f(x; W_0)\| \le L\epsilon$, but this can be conservative. FTI does not detect all forms of malicious behavior and should be combined with behavioral evaluation. Our prototype uses KZG-style commitments requiring a trusted setup; alternatives based on STARKs or inner-product arguments remove this assumption at the cost of larger proofs.

Several directions remain open. These include extending the framework to richer update patterns, such as mixed-precision updates and structured pruning; strengthening the link between parametric constraints and semantic behavior through task-specific sensitivity analysis; and developing standardized policy templates for common fine-tuning methods.

\bibliographystyle{splncs04}
\bibliography{refs}

\appendix

\section{Deferred Proofs}
\label{app:proofs}

\subsection{Proof of Theorem~\ref{thm:norm-soundness} (NBDP Soundness)}

Let $v = \mathrm{vec}(\Delta) \in \R^n$ where $n = dd'$, so $\|v\|_2 = \|\Delta\|_F$. Assume $\|\Delta\|_F > (1+\gamma)\epsilon$ where $\gamma = \sqrt{\log(2m/\delta)/\log(1/\delta)} - 1$.

\paragraph{Step 1: Single-projection failure probability.}
By Lemma~\ref{lem:hoe}, for each Rademacher projection $z_i = \langle r_i, v \rangle$,
\[
\Pr[|z_i| \le \tau] \le 2\exp\!\left(-\frac{\tau^2}{2\|v\|_2^2}\right) \le 2\exp\!\left(-\frac{\tau^2}{2(1+\gamma)^2\epsilon^2}\right).
\]
Substituting $\tau = \epsilon\sqrt{2\log(2m/\delta)}$ gives
\[
\Pr[|z_i| \le \tau] \le 2\exp\!\left(-\frac{\log(2m/\delta)}{(1+\gamma)^2}\right).
\]
By definition of $\gamma$, $(1+\gamma)^2 = \log(2m/\delta)/\log(1/\delta)$, so the exponent simplifies to $-\log(1/\delta)$, yielding $\Pr[|z_i| \le \tau] \le 2\delta$. Hence $\Pr[|z_i| > \tau] \ge 1 - 2\delta \ge c$ for a constant $c > 0$ when $\delta < 1/4$.

\paragraph{Step 2: Union over $m$ projections.}
The projections are independent, so the probability that all satisfy $|z_i| \le \tau$ is at most $(1 - c)^m$. For $m = \lceil 4\epsilon^{-2}\log(1/\delta) \rceil \ge c^{-1}\log(1/\delta)$, we have $(1-c)^m \le \exp(-cm) \le \delta$.

\paragraph{Step 3: Reduction to primitive soundness.}
The commitment scheme is binding, and the linear-relation and range proofs are sound. Therefore any accepting transcript corresponds to values $z_i$ consistent with the committed models that satisfy $|z_i| \le \tau$. An adversary can only succeed by either breaking one of these primitives (which happens with at most $\negl(\lambda)$ probability) or producing projections that all fall below $\tau$. The latter occurs with probability at most $\delta$ when $\|\Delta\|_F > (1+\gamma)\epsilon$. It follows that acceptance implies $\|\Delta\|_F \le (1+\gamma)\epsilon$ except with probability $\delta + \negl(\lambda)$.

\subsection{Proof of Theorem~\ref{thm:mrdp-soundness} (MRDP Soundness)}

Suppose $\operatorname{rank}(\Delta) > r$. The prover commits to univariate polynomials $f_1,\dots,f_r$ of degree at most $d-1$ and $g_1,\dots,g_r$ of degree at most $d'-1$. Define
\[
Q(X,Y) = P_\Delta(X,Y) - \sum_{k=1}^r f_k(X)\,g_k(Y).
\]

\paragraph{Step 1: $Q$ is nonzero.}
The polynomial $\sum_{k=1}^r f_k(X)\,g_k(Y)$ encodes a matrix of rank at most $r$: evaluating at the grid points $(X^i, Y^j)$ recovers a $d \times d'$ matrix that is a sum of $r$ outer products. Since $\operatorname{rank}(\Delta) > r$, the matrix $\Delta$ cannot be represented as such a sum, so $Q$ is not identically zero.

\paragraph{Step 2: Degree bound.}
Each monomial $X^i Y^j$ in $P_\Delta$ satisfies $i \le d-1$ and $j \le d'-1$, giving total degree at most $D = (d-1)+(d'-1) = d+d'-2$. Similarly, each $f_k(X)g_k(Y)$ has total degree at most $D$, so $\deg(Q) \le D$.

\paragraph{Step 3: Schwartz--Zippel application.}
The verifier samples $(x,y) \leftarrow \F_q^2$ uniformly. Since $Q$ is nonzero of degree $\le D$, $\Pr[Q(x,y)=0] \le D/q$.

\paragraph{Step 4: Field size requirement.}
For negligible soundness error, we need $D/q \le \negl(\lambda)$. With $D \le 2 \cdot 4096 = 8192$ for typical transformer blocks and $q \ge 2^{256}$ (e.g., the BLS12-381 scalar field), the error is at most $2^{13}/2^{256} < 2^{-243}$. No repetition is needed in this regime.

\paragraph{Step 5: Reduction to commitment soundness.}
Binding of the polynomial commitment ensures the prover cannot change the committed polynomials after observing the challenge $(x,y)$. Soundness of the opening proofs ensures that the opened values match the commitments. Any adversary that makes the verifier accept on a rank-$(r+1)$ or higher drift must either produce $Q(x,y)=0$ at the random point (probability $\le D/q$) or break the commitment scheme (probability $\le \negl(\lambda)$). The total soundness error is $D/q + \negl(\lambda)$.

\subsection{Proof of SDIP Soundness}

Assume $\|\Delta\|_{0} > k$. We show the verifier rejects with overwhelming probability.

\paragraph{Step 1: Discrepancy vector.}
Any claimed support $S$ with $|S| \le k$ must omit at least one nonzero coordinate of $\Delta$. Let $\Delta'$ be the prover's claimed sparse update restricted to $S$, and set $h = \Delta - \Delta'$. Since $\|\Delta\|_0 > k$ and $|S| \le k$, $h$ has at least one nonzero entry, so $h \neq 0$.

\paragraph{Step 2: Single-challenge detection.}
For a uniformly random $r \in \F_q^n$, Lemma~\ref{lem:sparse-inner} gives
$\Pr[\langle r, h \rangle = 0] = 1/q$.
The linear check computes $\langle r, W^* \rangle - \langle r, W_0 \rangle - \langle r, \Delta' \rangle = \langle r, h \rangle$, which is nonzero with probability $1 - 1/q$.

\paragraph{Step 3: Amplification via repetition.}
With $t$ independent challenges $r^{(1)},\dots,r^{(t)} \leftarrow \F_q^n$, the probability that all checks miss the discrepancy is $\Pr[\forall j: \langle r^{(j)}, h \rangle = 0] = q^{-t}$. For $q \ge 2^{256}$ and $t = 1$, this already gives negligible error. For smaller fields, choosing $t = \lceil \lambda / \log q \rceil$ ensures error $\le 2^{-\lambda}$.

\paragraph{Step 4: Reduction to primitive soundness.}
Binding of the commitment to $u$ and $\Delta'$ ensures the prover fixes its claimed support and values before seeing the challenges (in the Fiat--Shamir model, challenges are derived from the transcript including commitments). Soundness of the sub-protocols for the sum check ($\sum_i u_i \le k$) and consistency ($\Delta'_i = 0$ when $u_i = 0$) ensures that the committed $u$ and $\Delta'$ are well-formed. Any adversary that makes the verifier accept must either produce $\langle r^{(j)}, h \rangle = 0$ for all $j$ (probability $q^{-t}$) or break one of the primitives (probability $\negl(\lambda)$). Total soundness error: $q^{-t} + \negl(\lambda)$.

\section{Zero-Knowledge Proofs}
\label{app:zk-proofs}

All three protocols follow the same pattern: the simulator replaces real committed values with random ones drawn from the same domain, then invokes sub-protocol simulators for each proof component. Indistinguishability follows by a hybrid argument replacing one component at a time, reducing each step to hiding of $\Com$ or ZK of a sub-protocol.

\paragraph{NBDP (Theorem~\ref{thm:nbdp-zk}).}
$\mathcal{S}$ samples $\tilde{z}_i \leftarrow [-\tau, \tau]$ for $i \in [m]$, commits each, and invokes $\mathcal{S}_{\mathrm{lin}}$ and $\mathcal{S}_{\mathrm{rng}}$ to produce simulated proofs $\tilde{\pi}_i^{\mathrm{lin}}, \tilde{\pi}_i^{\mathrm{rng}}$. Hybrid $H_j$ uses real proofs for projections $1,\dots,j$ and simulated for the rest; $H_j \to H_{j-1}$ reduces to hiding of $\Com$ and ZK of sub-protocols. Total advantage: $m \cdot \negl(\lambda)$.

\paragraph{MRDP (Theorem~\ref{thm:mrdp-zk}).}
$\mathcal{S}$ samples random $(\tilde{f}_k, \tilde{g}_k)$ of appropriate degrees, commits, programs the Fiat--Shamir oracle to produce $(\tilde{x},\tilde{y})$, and uses the polynomial commitment simulator to open $P_{W^*} - P_{W_0}$ at $(\tilde{x},\tilde{y})$ consistently with $\sum_k \tilde{f}_k(\tilde{x})\tilde{g}_k(\tilde{y})$. Hybrid $H_1$: replace $2r$ commitments (hiding); $H_2$: replace opening proofs (ZK). Advantage: $(2r+1) \cdot \negl(\lambda)$.

\paragraph{SDIP (Theorem~\ref{thm:sdip-zk}).}
$\mathcal{S}$ samples random $\tilde{S} \subset [n]$ with $|\tilde{S}|=k$, random $\tilde{\Delta}'$ on $\tilde{S}$, commits $\tilde{u}$ and $\tilde{\Delta}'$, invokes $\mathcal{S}_{\mathrm{sum}}$, $\mathcal{S}_{\mathrm{cons}}$, and $\mathcal{S}_{\mathrm{lin}}$ for each of $t$ linear checks. Hybrid $H_1$: replace 2 commitments (hiding); $H_2$: replace $2+t$ sub-protocol proofs (ZK). Advantage: $(4+t) \cdot \negl(\lambda)$.

\end{document}